\documentclass[12pt, reqno]{amsart} 
\usepackage{amsmath, amsthm, amscd, amsfonts, amssymb, graphicx, color}
\usepackage[bookmarksnumbered, colorlinks, plainpages]{hyperref}

\newtheorem{theorem}{Theorem}[section]
\newtheorem{lemma}[theorem]{Lemma}
\newtheorem{proposition}[theorem]{Proposition}
\newtheorem{corollary}[theorem]{Corollary}
\theoremstyle{definition}
\newtheorem{definition}[theorem]{Definition}
\newtheorem{example}[theorem]{Example}

\theoremstyle{remark}
\newtheorem{remark}[theorem]{Remark}
\numberwithin{equation}{section}

\begin{document}
\title[On cyclic DNA codes over $\mathbb{F}_2+u\mathbb{F}_2+u^2\mathbb{F}_2$]{On cyclic DNA codes over $\mathbb{F}_2+u\mathbb{F}_2+u^2\mathbb{F}_2$}
\author[Hojjat Mostafanasab and Ahmad Yousefian Darani]{Hojjat Mostafanasab and Ahmad Yousefian Darani}


\begin{abstract}
In the present paper we study the structure of cyclic DNA codes of even length over the ring $\mathbb{F}_2+u\mathbb{F}_2+u^2\mathbb{F}_2$ where $u^3=0$. 
We investigate two presentations of cyclic codes of even length over 
$\mathbb{F}_2+u\mathbb{F}_2+u^2\mathbb{F}_2$ satisfying the reverse 
constraint and the reverse-complement constraint.
\end{abstract}

\maketitle

\section{Introduction}
Deoxyribonucleic acid (DNA) is a nucleic acid containing the genetic instructing used as the carrier of genetic information in all living organisms. DNA sequences consists of four bases nucleotides: Adenine ($A$), guanine ($G$), thymine ($T$) and cytosine ($C$). A single DNA
strand is an ordered quaternary sequence of the letters $A$, $G$, $T$, and $C$ with chemically distinct polar terminals known as the
5- and 3-ends. These strands are paired with each other as a double helix. This pairing is done by obeying the Watson-Crick
model. According to this model, $A$ and $T$ bound to each other and $G$ and $C$ bound to each other. $A$ and $G$ are called the
complements of $T$ and $C$, respectively or vice versa. The complement of a base say $X$ will be denoted by $\overline{X}$, for instance, the complement of $A$ is $\overline{A}=T$. As a general example, if $X=(GCATAG)$ is a DNA strand, then its complement is $\overline{X}=(CGTATC)$. DNA strand pairing is done in
the opposite direction and the reverse order. For instance, the Watson-Crick complementary
(WCC) strand of $3'-ACTTAGA-5'$ is the strand $5'-TCTAAGT -3'$. 

The interest on DNA computing started by the pioneer paper written by Leonard Adleman \cite{la}. Adleman
solved a hard (NP-complete) computational problem by DNA molecules in a test tube. 
Also, Richard J. Lipton in \cite{rj} solved a hard computational problem using the structure of DNA. 

Cyclic codes over finite rings played a very important role in the area of error correcting codes, see  
\cite{ao,absi,d,vz,bn,yk}. Since then, the structure of DNA is used as a model for constructing good error correcting codes and
conversely error correcting codes that enjoy similar properties with DNA structure are also
used to understand DNA itself. Gaborit and King in \cite{pd} discussed linear construction of DNA codes. In \cite{tax}, DNA codes over finite fields with four elements were studied by Abualrub et al. Later, Siap et al. studied DNA codes over the finite ring $\mathbb{F}_2[u]/<u^2-1>$ with four elements in \cite{ita}. Yildiz and Siap in \cite{YS} studied DNA codes over the ring $\mathbb{F}_2[u]/<u^4-1>$ with 16 elements. Some other papers devoted to cyclic DNA codes, see 
\cite{aei,ktp,ktp2,OSB}.

Liang and Wang \cite{LW}, investigated cyclic DNA codes over the ring $\mathbb{F}_2+u\mathbb{F}_2$ where $u^2=0$.
Cyclic codes over the ring $\mathbb{F}_2+u\mathbb{F}_2+u^2\mathbb{F}_2$ have been discussed in series of paper \cite{absi,Ceng,Cen,QZZ,SK}. Here, we study a family of cyclic DNA  codes of the finite ring $\mathbb{F}_2+u\mathbb{F}_2+u^2\mathbb{F}_2$ with 8 elements. Mainly, we discuss on cyclic codes of even length over $\mathbb{F}_2+u\mathbb{F}_2+u^2\mathbb{F}_2$ satisfying the reverse constraint and the reverse-complement constraint. 
Also, we determine the relationship between the generators of cyclic codes over $\mathbb{F}_2+u\mathbb{F}_2+u^2\mathbb{F}_2$ and their duals.

\section{Preliminaries}
A {\it linear code} $\mathcal{C}$ of length $n$ over a commutative ring $R$ is an $R$-submodule of $R^n$. 
An element of $\mathcal{C}$ is called a {\it codeword}.
A code of length $n$ is {\it cyclic} if the code is invariant under the automorphism $\sigma$ with
$\sigma(c_0, c_1, \cdots , c_{n-1}) = (c_{n-1}, c_0, \cdots , c_{n-2})$.
It is well known that a cyclic code of length $n$ over $R$ can be identified with an ideal in the quotient ring $R[x]/\langle x^n-1\rangle$ via the $R$-module isomorphism
as follows:
\begin{center}
	$R^n \longrightarrow R[x]/\langle x^n-1\rangle$\\
	$(c_0,c_1,\dots,c_{n-1}) \mapsto c_0+c_1x+\cdots+c_{n-1}x^{n-1}~ (\text{mod}~ \langle x^n-1\rangle) $
\end{center}

There are 16 pairs constructed by four basic nucleotides $A,T,G ~\text{and}~C$ such as 
\begin{center}
	$AA,TT,GG,CC,AT,TA,GC,CG,GT,TG,AC,CA,CT,TC,AG,GA$.
\end{center}

Let $\mathcal{R}$ be the ring $\mathbb{F}_2+u\mathbb{F}_2+u^2\mathbb{F}_2=\{0, 1, u, u^2, 1+u, 1+u^2, 1+u+u^2, u+u^2\}$ where $u^3=0$ mod 2. 
Since the ring $\mathcal{R}$ is of the cardinality 8, then we define the map $\Phi$ which gives a one-to-one correspondence between the elements of $\mathcal{R}$ and the 8 codons  which are given in Table 1. 
\begin{center}
	{\bf Table 1.} Identifying 8 codons with the elements of the ring $\mathcal{R}$.\\~\\
	\begin{tabular}{ l  c  c  c  c  c  c  c  c }
		\hline
		
		$GC$ & $0$ & $CG$ & $u^2$  \\
		$AT$ & $1$ & $TA$ & $1+u^2$  \\
		$GT$ & $u$ & $CA$ & $u+u^2$  \\
		$TG$ & $1+u$ & $AC$ & $1+u+u^2$ \\
		\hline
		
	\end{tabular}
\end{center}

Let $X=x_0x_1 \cdots x_{n-1} \in \mathcal{R}^n$ be a vector. The reverse of $X$ is defined as $X^r = x_{n-1}x_{n-2} \cdots x_1x_0$, the complement of $X$
is $X^c =\overline{x_0}~\overline{x_1}\cdots \overline{x_{n-1}} $, and the reverse-complement, also called the Watson-Crick complement (WCC) is defined as $X^{rc} =\overline{x_{n-1}}~\overline{x_{n-2}} \cdots \overline{x_1}~\overline{x_0} $.
\begin{definition}
A linear code $\mathcal{C}$ of length $n$ over $\mathcal{R}$ is said to be {\it reversible} if $X^r \in \mathcal{C}$ for all $X \in \mathcal{C}$, {\it complement} if $X^c \in \mathcal{C}$ for all $X \in \mathcal{C}$ and {\it reversible-complement} if $X^{rc} \in \mathcal{C}$ for all $X\in\mathcal{C}$.	
\end{definition}

\begin{definition}
	A linear code $\mathcal{C}$ of length $n$ over $\mathcal{R}$ is called a {\it DNA code} if
	\begin{enumerate}
		\item $\mathcal{C}$ is a cyclic code, i.e. $\mathcal{C}$ is an ideal of $\mathcal{R}[x]/\langle x^n-1\rangle$; and
		\item For any codeword $X\in \mathcal{C}, X\neq X^{rc}$	and $X^{rc} \in \mathcal{C}$.
	\end{enumerate}
\end{definition}

\begin{theorem}$($\cite[Theorem 2(2)]{absi}$)$
Let $\mathcal{C}$ be a cyclic code of even length $n$ over $\mathcal{R}$. Then
\begin{enumerate}
\item $\mathcal{C}=\langle g(x)+up_1(x)+u^2p_2(x)\rangle$ 
where $g(x),p_i(x)\in\mathbb{F}_2[x]$ with $g(x)\mid(x^n-1)$ mod 2, $\big(g(x)+up_1(x)\big)\mid(x^n-1)$
in $\mathbb{F}_2+u\mathbb{F}_2$ and $\big(g(x)+up_1(x)+u^2p_2(x)\big)\mid(x^n-1)$ in $\mathcal{R}$ and $deg(p_2(x))<deg(p_1(x))$. Or
\item $\mathcal{C}=\langle g(x)+up_1(x)+u^2p_2(x), u^2a_2(x)\rangle$ 
with $a_2(x)\mid g(x)\mid(x^n-1)$ mod 2 and $\big(g(x)+up_1(x)\big)\mid(x^n-1)$
in $\mathbb{F}_2+u\mathbb{F}_2$. Moreover, $g(x)\mid p_1(x)\left(\frac{x^n-1}{g(x)}\right)$, 
$a_2(x)\mid p_1(x)\left(\frac{x^n-1}{g(x)}\right)$, $a_2(x)\mid p_2(x)\left(\frac{x^n-1}{g(x)}\right)^2$
and $deg(p_2(x))<deg(a_2(x))$. Or
\item $\mathcal{C}=\langle g(x)+up_1(x)+u^2p_2(x), ua_1(x)+u^2q(x), u^2a_2(x)\rangle$ 
with $a_2(x)\mid a_1(x)\mid g(x)\mid(x^n-1)$ mod 2, $a_1(x)\mid p_1(x)\left(\frac{x^n-1}{g(x)}\right)$, 
$a_2(x)\mid q(x)\left(\frac{x^n-1}{a_1(x)}\right)$ and $a_2(x)\mid p_2(x)\left(\frac{x^n-1}{g(x)}\right)\left(\frac{x^n-1}{a_1(x)}\right)$.
Moreover, $deg(p_1(x))<deg(a_1(x))$, $deg(p_2(x))<deg(a_2(x))$ and $deg(q(x))<deg(a_2(x))$.
\end{enumerate}
\end{theorem}

For each polynomial $f(x)=a_0+a_1x+\cdots+a_rx^r$ with $a_r\neq 0$, {\it the reciprocal of} $f(x)$ is the polynomial
$f^*(x)=x^rf(1/x)=a_r+a_{r-1}x+\cdots+a_0x^r.$
It is easy to see that $deg(f ^*(x)) \leq deg(f (x))$ and if $a_0 \neq 0$, then $deg(f ^*(x)) = deg(f (x))$. 
A polynomial $f(x)$ is called {\it self-reciprocal} if $f^*(x) = f(x)$. It is easy to check that if $f(x),g(x)$ are two 
polynomials such that $f(x)\mid g(x)$, then $f^*(x)\mid g^*(x)$.

\begin{lemma}\label{star}
Let $R$ be a commutative ring and $f(x)$ be a polynomial in $R[x]/\langle x^n-1\rangle$ with $deg(f(x))=t$. Then
\begin{enumerate}
\item  $x^{t+1}f(x)^r=f^*(x)$ in $R[x]/\langle x^n-1\rangle$.
\item $x^{n-t-1}f^*(x)=f(x)^r$ in $R[x]/\langle x^n-1\rangle$.
\end{enumerate}
\end{lemma}
\begin{proof}
(1) Assume that $f(x)=f_0+f_1x+\dots+f_tx^t$. Then $$f(x)^r=f_tx^{n-t-1}+\dots+f_1x^{n-2}+f_0x^{n-1}.$$
Hence $$x^{t+1}f(x)^r=f_t+f_{t-1}x+\dots+f_1x^{t-1}+f_0x^{t}=f^*(x).$$
(2) By part (1).
\end{proof}

\begin{proposition}
Let $\mathcal{C}$ be a cyclic code of lenght $n$ over a commutative ring $R$ and $f(x)\in R[x]/\langle x^n-1\rangle$. Then $f(x)^r\in\mathcal{C}$ if and only if $f^*(x)\in\mathcal{C}$.
\end{proposition}
\begin{proof}
Use Lemma \ref{star}.
\end{proof}

As a direct consequence of the above proposition, the following result follows:
\begin{corollary}
Let $\mathcal{C}$ be a linear code. Then $\mathcal{C}$ is reversible if and only if $f^*(x)\in\mathcal{C}$ for all $f(x)\in\mathcal{C}$.
\end{corollary}

\begin{lemma}\label{equality}
Let $f_i(x),g_i(x)\in \mathbb{F}_2[x]$ for $i=1,2,3$. If $$f_1(x)+uf_2(x)+u^2f_3(x)=g_1(x)+ug_2(x)+u^2g_3(x),$$
then $f_i(x)=g_i(x)$ for $i=1,2,3$.
\end{lemma}
\begin{proof}
Let $f_1(x)+uf_2(x)+u^2f_3(x)=g_1(x)+ug_2(x)+u^2g_3(x)$. Multiplying $u^2$ on both sides we get 
$u^2f_1(x)=u^2g_1(x)$. So $f_1(x)=g_1(x)$. Hence $uf_2(x)+u^2f_3(x)=ug_2(x)+u^2g_3(x)$.
Thus $u^2f_2(x)=u^2g_2(x)$, which shows that $f_2(x)=g_2(x)$. Finally, easily we have that $f_3(x)=g_3(x)$.
\end{proof}

\begin{definition}
Let $\mathcal{C}$ be a cyclic code over $\mathcal{R}$. Define 
$$\mathcal{C}_{u^2}=\{u^2k(x)\mid  k(x)\in\mathbb{F}_2[x] \mbox{ and } u^2k(x)\in\mathcal{C}\}.$$
Notice that $\mathcal{C}_{u^2}$ is a subcode of $\mathcal{C}$.
\end{definition}

\begin{remark}
Regarding the map $\Phi$, we have that $\Phi(u^2\mathcal{R})=\{GC,CG\}$. 
\end{remark}

\begin{theorem} 
Let $\mathcal{C}=\langle g(x)+up_1(x)+u^2p_2(x), u^2a_2(x)\rangle$ be a cyclic code over $\mathcal{R}$ with
$a_2(x)\mid g(x)$ mod 2. Then $\mathcal{C}_{u^2}=\langle u^2a_2(x)\rangle$.
\end{theorem}
\begin{proof}
Clearly $\langle u^2a_2(x)\rangle\subseteq\mathcal{C}_{u^2}$. Let $u^2k(x)\in\mathcal{C}$ be such that $k(x)\in\mathbb{F}_2[x]$.
Then, there are two polynomials $\lambda(x),\mu(x)\in\mathcal{R}[x]$ such that
\begin{eqnarray*}
u^2k(x)&=&\lambda(x)\big(g(x)+up_1(x)+u^2p_2(x)\big)+u^2\mu(x)a_2(x)\in\mathcal{C}.
\end{eqnarray*} 
Let $\lambda(x)=\lambda_0(x)+u\lambda_1(x)+u^2\lambda_2(x)$, where $\lambda_i(x)$'s are polynomials in $\mathbb{F}_2[x]$. Note that, we may assume that $\mu(x)\in\mathbb{F}_2[x]$. Thus
\begin{eqnarray*}
u^2k(x)
&=&\lambda_0(x)g(x)+u\big(\lambda_0(x)p_1(x)+\lambda_1(x)g(x)\big)\\
&+&u^2(\lambda_0(x)p_2(x)+\lambda_1(x)p_1(x)+\lambda_2(x)g(x)+\mu(x)a_2(x)\big).
\end{eqnarray*}
Therefore, $\lambda_0(x)=\lambda_1(x)=0$ and $k(x)=\lambda_2(x)g(x)+\mu(x)a_2(x)$. Since $a_2(x)\mid g(x)$,
then $a_2(x)\mid k(x)$ and so $u^2k(x)\in\langle u^2a_2(x)\rangle$. Consequently
$\mathcal{C}_{u^2}=\langle u^2a_2(x)\rangle$.
\end{proof}

Now, we give a useful lemma. 
\begin{lemma}\label{equalityreciprocal}
Let $f_i(x)\in\mathbb{F}_2[x]$ for $i=1,2,3$. Suppose that $deg(f_1(x))=r$, $deg(f_2(x))=s$ and $deg(f_3(x))=t$ 
where $r>max\{s,t\}$. Then $$\big(f_1(x)+uf_2(x)+u^2f_3(x)\big)^*=f_1^*(x)+ux^{r-s}f_2^*(x)+u^2x^{r-t}f_3^*(x).$$
\end{lemma}
\begin{proof}
Without loss of the generality we may assume that $s\geq t$. By \cite[Lemma 4.2]{YS}, we deduce that 
\begin{eqnarray*}
\big(f_1(x)+uf_2(x)+u^2f_3(x)\big)^*&=&
f_1^*(x)+x^{r-s}(uf_2(x)+u^2f_3(x))^*\\
&=&f_1^*(x)+x^{r-s}(uf_2^*(x)+x^{s-t}u^2f_3^*(x))\\
&=&f_1^*(x)+ux^{r-s}f_2^*(x)+u^2x^{r-t}f_3^*(x).
\end{eqnarray*}
\end{proof}

\section{The reverse constraint}

\begin{theorem}\label{firstpresentation}
Let $\mathcal{C}=\langle g(x)+up_1(x)+u^2p_2(x)\rangle$ be a cyclic code of even length $n$ over $\mathcal{R}$ with $deg(g(x))=r$, $deg(p_1(x))=s$ and $deg(p_2(x))=t$ where $r>s$. Then $\mathcal{C}$ is reversible if and only if 
\begin{enumerate}
\item $g(x)$ is self-reciprocal;
\item \hspace{.4cm}$(a)$ $x^{r-s}p_1^*(x)=p_1(x)$ and $x^{r-t}p_2^*(x)=p_2(x)$, or

$(b)$ $x^{r-s}p_1^*(x)=g(x)+p_1(x)$ and $x^{r-t}p_2^*(x)=p_1(x)+p_2(x)$, or

$(c)$ $x^{r-s}p_1^*(x)=p_1(x)$ and $x^{r-t}p_2^*(x)=g(x)+p_2(x)$, or

$(d)$ $x^{r-s}p_1^*(x)=g(x)+p_1(x)$ and $x^{r-t}p_2^*(x)=g(x)+p_1(x)+p_2(x)$.
\end{enumerate}
\end{theorem}
\begin{proof}
$(\Rightarrow)$ Suppose that $\mathcal{C}$ is reversible. Consider $\mathcal{C}$ as an $\mathcal{R}[x]$-module.
Then, by Lemma \ref{equalityreciprocal}, it follows that
\begin{eqnarray*}
\big(g(x)+up_1(x)+u^2p_2(x)\big)^*&=&
g^*(x)+ux^{r-s}p_1^*(x)+u^2x^{r-t}p_2^*(x)\\
&=&k(x)\big(g(x)+up_1(x)+u^2p_2(x)\big)\in\mathcal{C},
\end{eqnarray*}
for some $k(x)\in\mathcal{R}[x]$. Assume that $k(x)=k_0(x)+uk_1(x)+u^2k_2(x)$ where $k_i(x)$'s 
are polynomials in $\mathbb{F}_2[x]$.
This implies 
\begin{eqnarray*}
g^*(x)+ux^{r-s}p_1^*(x)+u^2x^{r-t}p_2^*(x)
&=&k_0(x)g(x)+u\big(k_0(x)p_1(x)+k_1(x)g(x)\big)\\
&+&u^2(k_0(x)p_2(x)+k_1(x)p_1(x)+k_2(x)g(x)\big).
\end{eqnarray*}
Now, Lemma \ref{equality} implies that 
$g^*(x)=k_0(x)g(x)$, $x^{r-s}p_1^*(x)=k_0(x)p_1(x)+k_1(x)g(x)$ and $x^{r-t}p_2^*(x)=k_0(x)p_2(x)+k_1(x)p_1(x)+k_2(x)g(x)$.
Since $g^*(x)=k_0(x)g(x)$ and $deg(g^*(x)) \leq deg(g(x))$, we have that $k_0(x)=1$ and so $g(x)$ is self-reciprocal.
Therefore $x^{r-s}p_1^*(x)=p_1(x)+k_1(x)g(x)$; whence, comparing the degrees of the two sides of this equality, shows that
$k_1(x)=0$ or 1. On the other hand $x^{r-t}p_2^*(x)=p_2(x)+k_1(x)p_1(x)+k_2(x)g(x)$. 
Again, comparing the degrees of the two sides, implies that
$k_2(x)=0$ or 1, and so 
there are four cases:\\
{\bf Case 1.} If $k_1(x)=k_2(x)=0$, then $x^{r-s}p_1^*(x)=p_1(x)$ and $x^{r-t}p_2^*(x)=p_2(x)$.\\
{\bf Case 2.} If $k_1(x)=1,~k_2(x)=0$, then $x^{r-s}p_1^*(x)=p_1(x)+g(x)$ and $x^{r-t}p_2^*(x)=p_2(x)+p_1(x)$.\\
{\bf Case 3.} If $k_1(x)=0,~k_2(x)=1$, then $x^{r-s}p_1^*(x)=p_1(x)$ and $x^{r-t}p_2^*(x)=p_2(x)+g(x)$.\\
{\bf Case 4.} If $k_1(x)=k_2(x)=1$, then $x^{r-s}p_1^*(x)=p_1(x)+g(x)$ and $x^{r-t}p_2^*(x)=p_2(x)+p_1(x)+g(x)$.\\
$(\Leftarrow)$ Assume that (1) holds. If (2)(a) holds, then
\begin{eqnarray*}
\big(g(x)+up_1(x)+u^2p_2(x)\big)^*&=&
g^*(x)+ux^{r-s}p_1^*(x)+u^2x^{r-t}p_2^*(x)\\
&=&g(x)+up_1(x)+u^2p_2(x)\in\mathcal{C}.
\end{eqnarray*}
In the case (2)(b) holds, we have
\begin{eqnarray*}
\big(g(x)+up_1(x)+u^2p_2(x)\big)^*&=&
g^*(x)+ux^{r-s}p_1^*(x)+u^2x^{r-t}p_2^*(x)\\
&=&g(x)+u(p_1(x)+g(x))+u^2(p_1(x)+p_2(x))\\
&=&g(x)+up_1(x)+u^2p_2(x)+ug(x)+u^2p_1(x)\\
&=&g(x)+up_1(x)+u^2p_2(x)+u(g(x)+up_1(x)\\
&+&u^2p_2(x))\in\mathcal{C}.
\end{eqnarray*}
When (2)(c) holds, then
\begin{eqnarray*}
\big(g(x)+up_1(x)+u^2p_2(x)\big)^*&=&
g^*(x)+ux^{r-s}p_1^*(x)+u^2x^{r-t}p_2^*(x)\\
&=&g(x)+up_1(x)+u^2(p_2(x)+g(x))\\
&=&g(x)+up_1(x)+u^2p_2(x)+u^2g(x)\\
&=&g(x)+up_1(x)+u^2p_2(x)+u^2(g(x)\\
&+&up_1(x)+u^2p_2(x))\in\mathcal{C}.
\end{eqnarray*}
Finally, if we have (2)(d), then
\begin{eqnarray*}
\big(g(x)+up_1(x)+u^2p_2(x)\big)^*&=&
g^*(x)+ux^{r-s}p_1^*(x)+u^2x^{r-t}p_2^*(x)\\
&=&g(x)+u(p_1(x)+g(x))+u^2(p_1(x)+p_2(x)+g(x))\\
&=&g(x)+up_1(x)+u^2p_2(x)+ug(x)+u^2p_1(x)+u^2g(x)\\
&=&g(x)+up_1(x)+u^2p_2(x)+u(g(x)+up_1(x)\\
&+&u^2p_2(x))+u^2(g(x)+up_1(x)+u^2p_2(x))\in\mathcal{C}.
\end{eqnarray*}
Therefore $\mathcal{C}$ is reversible.
\end{proof}

\begin{theorem} \label{secondpresentation}
Let $\mathcal{C}=\langle g(x)+up_1(x)+u^2p_2(x),u^2a_2(x)\rangle$ be a cyclic code of even length $n$ over $\mathcal{R}$ with $deg(g(x))=r$, $deg(p_1(x))=s$ and $deg(p_2(x))=t$ where $r>max\{s,t\}$. Also, assume that $a_2(x)\mid g(x)\mid(x^n-1)$ mod 2. Then $\mathcal{C}$ is reversible if and only if 
\begin{enumerate}
\item $g(x)$ and $a_2(x)$ are self-reciprocal;
\item \hspace{.4cm}$(a)$ $x^{r-s}p_1^*(x)=p_1(x)$ and $a_2(x)\mid x^{r-t}p_2^*(x)+p_2(x)$, or

$(b)$ $x^{r-s}p_1^*(x)=g(x)+p_1(x)$ and $a_2(x)\mid x^{r-t}p_2^*(x)+p_1(x)+p_2(x)$.
\end{enumerate}
\end{theorem}
\begin{proof}
$(\Rightarrow)$
Assume that $\mathcal{C}$ is reversible.  Notice that there are two polynomials $t(x),s(x)\in\mathcal{R}[x]$ such that
\begin{eqnarray*}
\big(g(x)+up_1(x)+u^2p_2(x)\big)^*&=&
g^*(x)+ux^{r-s}p_1^*(x)+u^2x^{r-t}p_2^*(x)\\
&=&g(x)+ux^{r-s}p_1^*(x)+u^2x^{r-t}p_2^*(x)\\
&=&t(x)\big(g(x)+up_1(x)+u^2p_2(x)\big)+u^2s(x)a_2(x)\in\mathcal{C}.
\end{eqnarray*} 
Let $t(x)=t_0(x)+ut_1(x)+u^2t_2(x)$, where $t_i(x)$'s are polynomials in $\mathbb{F}_2[x]$. Clearly, we may assume that $s(x)\in\mathbb{F}_2[x]$. Thus
\begin{eqnarray*}
g^*(x)+ux^{r-s}p_1^*(x)+u^2x^{r-t}p_2^*(x)
&=&t_0(x)g(x)+u\big(t_0(x)p_1(x)+t_1(x)g(x)\big)\\
&+&u^2(t_0(x)p_2(x)+t_1(x)p_1(x)+t_2(x)g(x)\\
&+&s(x)a_2(x)\big).
\end{eqnarray*}
Similar to the proof of Theorem \ref{firstpresentation}, we deduce that $g(x)$ is self-reciprocal, $t_0(x)=1$
and $t_1(x)=0$ or $1$. If $t_1(x)=0$, then $x^{r-s}p_1^*(x)=p_1(x)$ and 
$x^{r-t}p_2^*(x)=p_2(x)+t_2(x)g(x)+s(x)a_2(x)$.
If $t_1(x)=1$, then $x^{r-s}p_1^*(x)=p_1(x)+g(x)$ and $x^{r-t}p_2^*(x)=p_2(x)+p_1(x)+t_2(x)g(x)+s(x)a_2(x)$.\\
By assumption we have $a_2(x)\mid g(x)$ mod 2. So

 $x^{r-s}p_1^*(x)=p_1(x)$ and $a_2(x)\mid x^{r-t}p_2^*(x)+p_2(x)$, or

 $x^{r-s}p_1^*(x)=g(x)+p_1(x)$ and $a_2(x)\mid x^{r-t}p_2^*(x)+p_1(x)+p_2(x)$.\\
Moreover $u^2a_2^*(x)\in\mathcal{C}$. Then, there exist polynomials $\lambda_0(x),\lambda_1(x),\lambda_2(x)\in\mathcal{R}[x]$
and $\mu(x)\in\mathbb{F}_2[x]$
such that 
\begin{eqnarray*}
u^2a_2^*(x)
&=&\lambda_0(x)g(x)+u\big(\lambda_0(x)p_1(x)+\lambda_1(x)g(x)\big)\\
&+&u^2(\lambda_0(x)p_2(x)+\lambda_1(x)p_1(x)+\lambda_2(x)g(x)+\mu(x)a_2(x)\big).
\end{eqnarray*}
Therefore, $\lambda_0(x)=\lambda_1(x)=0$ and $a_2^*(x)=\lambda_2(x)g(x)+\mu(x)a_2(x)$. Since $a_2(x)\mid g(x)$,
then $a_2(x)\mid a_2^*(x)$ and so $a_2^*(x)=a_2(x)$, i.e., $a_2(x)$ is self-reciprocal.\\
$(\Leftarrow)$ We investigate only one case. Suppose that (1) and part (b) of (2) hold. Then
\begin{eqnarray*}
\big(g(x)+up_1(x)+u^2p_2(x)\big)^*&=&
g^*(x)+ux^{r-s}p_1^*(x)+u^2x^{r-t}p_2^*(x)\\
&=&g(x)+u(p_1(x)+g(x))+u^2(p_1(x)+p_2(x)\\
&+&\lambda(x)a_2(x))~~~ \mbox{for some}~~ \lambda(x)\in\mathbb{F}_2[x]\\
&=&g(x)+up_1(x)+u^2p_2(x)+ug(x)\\
&+&u^2p_1(x)+\lambda(x)u^2a_2(x)\\
&=&g(x)+up_1(x)+u^2p_2(x)+u(g(x)+up_1(x)+u^2p_2(x))\\
&+&\lambda(x)u^2a_2(x)\in\mathcal{C}.
\end{eqnarray*}
Consequently $\mathcal{C}$ is reversible.
\end{proof}


\section{The reverse-complement constraint}
In this section, the reverse-complement constraint is examined for cyclic codes.

\begin{lemma}\label{tablenotes}
The following conditions hold:
\begin{enumerate}
\item For any elements $x,y,z\in\mathbb{F}_2$, $\overline{x+yu+zu^2}=\overline{x}+yu+zu^2$.
\item For any $a\in\mathcal{R}$, $a+\overline{a}=u^2$.
\item For any $a,b\in\mathcal{R}$, $\overline{a+b}=\overline{a}+\overline{b}+u^2$. 
\end{enumerate}
\end{lemma}
\begin{proof}
(1), (2) Regarding Table 1, the computations are easy.\\
(3) Let $a,b\in\mathcal{R}$. By part (2), 
$$\overline{a+b}=a+b+u^2=(\overline{a}+u^2)+(\overline{b}+u^2)+u^2=\overline{a}+\overline{b}+u^2.$$
\end{proof}

From now on, we denote $\mathbb{I}(x)=\frac{1+x^n}{1+x}$. 

\begin{theorem}\label{mainequality}
Let $f(x)\in\mathcal{R}[x]$. Then $f(x)^{rc}+u^2\mathbb{I}(x)=f(x)^r$.
\end{theorem}
\begin{proof}
Assume that $f(x)=f_0+f_1x+\dots+f_{r-1}x^{r-1}+x^r$. Then
\begin{eqnarray*}
f(x)^{rc}&=&\overline{0}+\overline{0}x+\dots+\overline{0}x^{n-r-2}+\overline{1}x^{n-r-1}\\
&+&\overline{f_{r-1}}x^{n-r}+\dots+\overline{f_1}x^{n-2}+\overline{f_0}x^{n-1}\\
&=&u^2+u^2x+\dots+u^2x^{n-r-2}+(1+u^2)x^{n-r-1}\\
&+&\overline{f_{r-1}}x^{n-r}+\dots+\overline{f_1}x^{n-2}+\overline{f_0}x^{n-1}\\
&=&u^2+u^2x+\dots+u^2x^{n-r-2}+(1+u^2)x^{n-r-1}\\
&+&(f_{r-1}+u^2)x^{n-r}+\dots+(f_1+u^2)x^{n-2}+(f_0+u^2)x^{n-1}.
\end{eqnarray*}
Hence
\begin{eqnarray*}
f(x)^{rc}+u^2\mathbb{I}(x)&=&x^{n-r-1}
+f_{r-1}x^{n-r}+\dots+f_1x^{n-2}+f_0x^{n-1}\\
&=&f(x)^r.
\end{eqnarray*}
\end{proof}

\begin{theorem}\label{rev-com1}
Let $\mathcal{C}$ be a cyclic code of lenght $n$ over $\mathcal{R}$. Then $\mathcal{C}$ is reversible-complement if and only if $\mathcal{C}$ is reversible and $u^2\mathbb{I}(x)\in\mathcal{C}$.
\end{theorem}
\begin{proof}
Let $\mathcal{C}$ be reversible-complement. Since $0+0x+\dots+0x^{n-1}\in\mathcal{C}$, we have that
$$\overline{0+0x+\dots+0x^{n-1}}=u^2(1+x+\dots+x^{n-1})=u^2\mathbb{I}(x)\in\mathcal{C}.$$ Now, apply Theorem \ref{mainequality}.
\end{proof}

\begin{theorem}\label{rev-com1}
Let $\mathcal{C}_1,\mathcal{C}_2$ be two reversible-complement cyclic codes of lenght $n$ over $\mathcal{R}$. Then $\mathcal{C}_1+\mathcal{C}_2$ and $\mathcal{C}_1\cap\mathcal{C}_2$ are reversible-complement cyclic codes.
\end{theorem}
\begin{proof}
Let $X=x_0x_1\cdots x_{n-1}$ and $Y=y_0y_1\cdots y_{n-1}$ be two arbitrary codewords in 
$\mathcal{C}_1$ and $\mathcal{C}_2$, respectively. Then 
\begin{eqnarray*}
(X+Y)^{rc}&=&\overline{(x_{n-1}+y_{n-1})}\cdots\overline{(x_1+y_1)}~\overline{(x_0+y_0)}\\
&=&(\overline{x_{n-1}}+\overline{y_{n-1}}+u^2)\cdots(\overline{x_1}+\overline{y_{1}}
+u^2)~(\overline{x_0}+\overline{y_0}+u^2)~~~ \mbox{By Lemma \ref{tablenotes}(3)}\\
&=&\big[(\overline{x_{n-1}}+u^2)\cdots(\overline{x_1}+u^2)~(\overline{x_0}+u^2)\big]
+\overline{y_{n-1}}\cdots\overline{y_1}~\overline{y_0}\\
&=&(X^{rc}+u^2\mathbb{I}(x))+Y^{rc}\in\mathcal{C}_1+\mathcal{C}_2
\end{eqnarray*}
which shows that $\mathcal{C}_1+\mathcal{C}_2$ is reversible-complement. It is evident that
$\mathcal{C}_1\cap\mathcal{C}_2$ is reversible-complement.
\end{proof}

\begin{corollary}
Let $\mathcal{C}=\langle g(x)+up_1(x)+u^2p_2(x)\rangle$ be a cyclic code of even length $n$ over $\mathcal{R}$ with $deg(g(x))=r$, $deg(p_1(x))=s$ and $deg(p_2(x))=t$ where $r>s$. Then $\mathcal{C}$ is reversible-complement if and only if 
\begin{enumerate}
\item $u^2\mathbb{I}(x)\in\mathcal{C}$.
\item $g(x)$ is self-reciprocal;
\item \hspace{.4cm}$(a)$ $x^{r-s}p_1^*(x)=p_1(x)$ and $x^{r-t}p_2^*(x)=p_2(x)$, or

$(b)$ $x^{r-s}p_1^*(x)=g(x)+p_1(x)$ and $x^{r-t}p_2^*(x)=p_1(x)+p_2(x)$, or

$(c)$ $x^{r-s}p_1^*(x)=p_1(x)$ and $x^{r-t}p_2^*(x)=g(x)+p_2(x)$, or

$(d)$ $x^{r-s}p_1^*(x)=g(x)+p_1(x)$ and $x^{r-t}p_2^*(x)=g(x)+p_1(x)+p_2(x)$.
\end{enumerate}
\end{corollary}

\begin{example}
Let $x^8-1=(x+1)^8=g^8$ over $\mathbb{F}_2$. Let $\mathcal{C}=\langle g(x)+up_1(x)+u^2p_2(x)\rangle$,
where $g(x)=g^6$, $p_1(x)=x^5+x$ and $p_2(x)=x^4+x^2$. It is easy to check that 
$g(x)=x^6+x^4+x^2+1$ is self-reciprocal. Also,
$x^ip_1^{*}(x)=p_1(x)$ and $x^jp_2^{*}(x)=p_2(x)$ where $i = deg(g(x))-deg(p_1(x))$ and $j= deg(g(x))-deg(p_2(x))$. Then $\mathcal{C}$ is a cyclic DNA code. 
Here, we present an algorithm for generating the 28 codewords of this code which are given in Table 2. 
In the first step, find the presentation of elements $\alpha(g(x)+up_1(x)+u^2p_2(x))$ where $\alpha\in\mathcal{R}$,
and in the second step, find the presentation of all cyclic shifts of these elements.
Note that the element $u^2\mathbb{I}(x)=u^2(x+1)g(x)=u^2(x+1)(g(x)+up_1(x)+u^2p_2(x))$ has the presentation $CGCGCGCGCGCGCGCG$. It is observed that $\mathcal{C}$ is a cyclic code with minimum Hamming distance 4
\newpage
\begin{center}
	{\bf Table 2.} DNA codes of length 8 obtained from $\mathcal{C}=\langle g(x)+up_1(x)+u^2p_2(x)\rangle$.\\~\\
	\begin{tabular}{ l  c  c  c  c  c  c  c  c }
		\hline
		
		$GCGCGCGCGCGCGCGC$& & & & & & & & $GTCGGTGCGTCGGTGC$\\
		$CGCGCGCGCGCGCGCG$& & & & & & & & $GCGTCGGTGCGTCGGT$\\
		$ATGTTAGCTAGTATGC$& & & & & & & & $GTGCGTCGGTGCGTCG$\\
		$GCATGTTAGCTAGTAT$& & & & & & & & $CGGTGCGTCGGTGCGT$\\
	    $ATGCATGTTAGCTAGT$& & & & & & & & $TGCAACGCACCATGGC$\\
	    $GTATGCATGTTAGCTA$& & & & & & & & $GCTGCAACGCACCATG$\\
	    $TAGTATGCATGTTAGC$& & & & & & & & $TGGCTGCAACGCACCA$\\
	    $GCTAGTATGCATGTTA$& & & & & & & & $CATGGCTGCAACGCAC$\\
	    $TAGCTAGTATGCATGT$& & & & & & & & $ACCATGGCTGCAACGC$\\
	    $GTTAGCTAGTATGCAT$& & & & & & & & $CACGCAGCCACGCAGC$\\
	    $CAACGCACCATGGCTG$& & & & & & & & $GCCACGCAGCCACGCA$\\
	    $CGCAGCCACGCAGCCA$& & & & & & & & $CAGCCACGCAGCCACG$\\
	    $GCCGGCCGGCCGGCCG$& & & & & & & & $CGGCCGGCCGGCCGGC$\\
	    $GCACCATGGCTGCAAC$& & & & & & & & $ACGCACCATGGCTGCA$\\
		\hline
	\end{tabular}
\end{center}
\end{example}

\begin{corollary} 
Let $\mathcal{C}=\langle g(x)+up_1(x)+u^2p_2(x),u^2a_2(x)\rangle$ be a cyclic code of even length $n$ over $\mathcal{R}$ with $deg(g(x))=r$, $deg(p_1(x))=s$ and $deg(p_2(x))=t$ where $r>max\{s,t\}$. Also, assume that $a_2(x)\mid a_1(x)\mid g(x)\mid(x^n-1)$ mod 2. Then $\mathcal{C}$ is reversible-complement if and only if 
\begin{enumerate}
\item $u^2\mathbb{I}(x)\in\mathcal{C}$.
\item $g(x)$ and $a_2(x)$ are self-reciprocal;
\item \hspace{.4cm}$(a)$ $x^{r-s}p_1^*(x)=p_1(x)$ and $a_2(x)\mid x^{r-t}p_2^*(x)+p_2(x)$, or

$(b)$ $x^{r-s}p_1^*(x)=g(x)+p_1(x)$ and $a_2(x)\mid x^{r-t}p_2^*(x)+p_1(x)+p_2(x)$.
\end{enumerate}
\end{corollary}

\section{Dual cyclic codes over $\mathcal{R}$}

Let $X=x_0x_1\dots x_{n-1}$ and $Y=y_0y_1\dots y_{n-1}$ be
two elements of $\mathcal{R}^n$. Then we denote the Euclidean
inner product in $\mathcal{R}^n$ as 
$\langle X,Y\rangle_E=x_0y_0+x_1y_1+\cdots+ x_{n-1}y_{n-1}$.
Two elements $X$ and $Y$ are called orthogonal with respect to
the Euclidean inner product, if $\langle X,Y\rangle_E=0$.
The Hermitian inner product is defined as
$\langle X,Y\rangle_H=x_0\overline{y_0}+x_1\overline{y_1}+\cdots+ x_{n-1}\overline{y_{n-1}}$.
 Two elements $X$ and $Y$ are
called orthogonal with respect to the Hermitian
inner product, if $\langle X,Y\rangle_H=0$.
The dual code $\mathcal{C}^\bot$ with respect to the Euclidean inner product of $\mathcal{C}$ is defined as
$$\mathcal{C}^\bot=\{X\in\mathcal{R}^n\mid\langle X,Y\rangle_E=0 \mbox{ for all } Y\in\mathcal{C}\}.$$
The dual code $\mathcal{C}^{\bot H}$ with respect to the Hermitian inner product of $\mathcal{C}$ is defined as
$$\mathcal{C}^{\bot H}=\{X\in\mathcal{R}^n\mid\langle X,Y\rangle_H=0 \mbox{ for all } Y\in\mathcal{C}\}.$$
In the following theorem, the Euclidean dual and the
Hermitian dual of a cyclic code of lenght $n$ over $\mathcal{R}$ are considered. 
\begin{theorem}
Let $\mathcal{C}$ be a cyclic code of lenght $n$ over $\mathcal{R}$. 
The following conditions are equivalent:
\begin{enumerate}
\item $\mathcal{C}$ is reversible;
\item $\mathcal{C}^{{\bot}}$ is reversible;
\item $\mathcal{C}^{{\bot}H}$ is reversible.
\end{enumerate}
\end{theorem}
\begin{proof}
$(1)\Rightarrow(2)$ Let $X=x_0x_1\dots x_{n-1}$ be an arbitrary codeword in $\mathcal{C}^\bot$. We claim that $X^{r}\in\mathcal{C}^\bot$. Get a codeword $Y=y_0y_1\dots y_{n-1}\in\mathcal{C}$. Since $\mathcal{C}$
is reversible, then $Y^{r}={y_{n-1}}\dots{y_1}~{y_0}\in\mathcal{C}$. 
Therefore $$\langle X, Y^{r}\rangle_{E}=x_{n-1}{y_0}+x_{n-2}{y_1}+\dots+x_0{y_{n-1}}=\langle X^{r}, Y\rangle_E=0.$$
Consequently the claim holds.\\
$(1)\Rightarrow(3)$
Let $X=x_0x_1\dots x_{n-1}$ be an arbitrary codeword in $\mathcal{C}^{\bot H}$ and $Y=y_0y_1\dots y_{n-1}$ be a codeword in $\mathcal{C}$. Since $\mathcal{C}$
is reversible, then $Y^{r}={y_{n-1}}\dots{y_1}~{y_0}\in\mathcal{C}$. 
Therefore $$\langle X, Y^{r}\rangle_H=x_{n-1}\overline{y_0}+x_{n-2}\overline{y_1}+\dots+x_0\overline{y_{n-1}}=\langle X^{r}, Y\rangle_H=0.$$
Consequently $X^{r}\in\mathcal{C}^{\bot H}$, which shows that $\mathcal{C}^{{\bot}H}$ is reversible.\\
$(2)\Rightarrow(1)$ and $(3)\Rightarrow(1)$ By the fact that for every cyclic code $\mathcal{C}$ we have $(\mathcal{C}^{\bot})^\bot=\mathcal{C}$
and $(\mathcal{C}^{{\bot}H})^{{\bot}H}=\mathcal{C}$.
\end{proof}

\begin{theorem} 
Let $\mathcal{C}=\langle g(x)+up_1(x)+u^2p_2(x), ua_1(x)+u^2q(x), u^2a_2(x)\rangle$ be a cyclic code of even lenght $n$ over $\mathcal{R}$ with
$a_2(x)\mid a_1(x)\mid g(x)\mid(x^n-1)$ mod 2, also $a_1(x)\mid p_1(x)$,~ $a_2(x)\mid p_2(x)$ and $a_2(x)\mid q(x)$.
Moreover, assume that $deg(g(x))=r$, $deg(p_1(x))=s$ and $deg(p_2(x))=t$ where $r>max\{s,t\}$, also $deg(a_1(x))=l$, $deg(q(x))=m$  with $l>m$.
If $\mathcal{C}^\bot=\langle \widehat{g}(x)+u\widehat{p_1}(x)+u^2\widehat{p_2}(x), u\widehat{a_1}(x)+u^2\widehat{q}(x), u^2\widehat{a_2}(x)\rangle$ is the dual cyclic code of $\mathcal{C}$, then
\begin{enumerate}
\item $\widehat{g}(x)=\frac{x^n-1}{a_2^*(x)}\alpha(x)$ for some $\alpha(x)\in\mathbb{F}_2[x]$.
\item $\widehat{a_1}(x)=\frac{x^n-1}{a_1^*(x)}\beta(x)$ for some $\beta(x)\in\mathbb{F}_2[x]$.
\item $\widehat{a_2}(x)=\frac{x^n-1}{g^*(x)}\gamma(x)$ for some $\gamma(x)\in\mathbb{F}_2[x]$.
\item $\widehat{q}(x)=\frac{x^n-1}{g^*(x)}\mu(x)$ for some $\mu(x)\in\mathbb{F}_2[x]$.
\item $\widehat{p_1}(x)=\frac{x^n-1}{a_1^*(x)}\nu(x)$ for some $\nu(x)\in\mathbb{F}_2[x]$.
\item $\widehat{p_2}(x)=\frac{x^n-1}{g^*(x)}\eta(x)$ for some $\eta(x)\in\mathbb{F}_2[x]$.
\end{enumerate}
\end{theorem}
\begin{proof}
(1) Since $\widehat{g}(x)+u\widehat{p_1}(x)+u^2\widehat{p_2}(x)\in\mathcal{C}^\bot$ and
$u^2a_2(x)\in\mathcal{C}$, then by Lemma 4.4.8 of \cite{HP}, we have that
$(\widehat{g}(x)+u\widehat{p_1}(x)+u^2\widehat{p_2}(x))(u^2a_2^*(x))=u^2\widehat{g}(x)a_2^*(x)=0$
mod $\langle x^n-1\rangle$.
Thus $\widehat{g}(x)a_2^*(x)=0$ mod $\langle x^n-1\rangle$. Notice that $a(x)\mid(x^n-1)$ implies $a^*(x)\mid (x^n-1)^*=x^n-1$.  So $\widehat{g}(x)=\frac{x^n-1}{a_2^*(x)}\alpha(x)$ for some $\alpha(x)\in\mathbb{F}_2[x]$.

(2) Again, by Lemma 4.4.8 of \cite{HP},
$$(u\widehat{a_1}(x)+u^2\widehat{q}(x))(ua_1(x)+u^2q(x))^*=u^2\widehat{a_1}(x)a_1^*(x)=0
\mbox{ mod }\langle x^n-1\rangle.$$ Hence
$\widehat{a_1}(x)=\frac{x^n-1}{a_1^*(x)}\beta(x)$ 
for some $\beta(x)\in\mathbb{F}_2[x]$.

(3) Since $(u^2\widehat{a_2}(x))(g(x)+up_1(x)+u^2p_2(x))^*=u^2\widehat{a_2}(x)g^*(x)=0$ mod $\langle x^n-1\rangle$, then $\widehat{a_2}(x)=\frac{x^n-1}{g^*(x)}\gamma(x)$ for some $\gamma(x)\in\mathbb{F}_2[x]$.

(4) Notice that 
\begin{eqnarray*}
(u\widehat{a_1}(x)+u^2\widehat{q}(x))(g(x)+up_1(x)+u^2p_2(x))^*&=&
u\widehat{a_1}(x)g^*(x)+u^2\widehat{a_1}(x)x^{r-s}p_1^*(x)\\
&+&u^2\widehat{q}(x)g^*(x)=0 \mbox{ mod } \langle x^n-1\rangle.
\end{eqnarray*}
Since $a_1^*(x)\mid g^*(x)$ and $a_1^*(x)\mid p_1^*(x)$, by part (2) we get 
$u\widehat{a_1}(x)g^*(x)=u^2\widehat{a_1}(x)p_1^*(x)=0$ mod $\langle x^n-1\rangle$.
Hence $u^2\widehat{q}(x)g^*(x)=0$ mod $\langle x^n-1\rangle$, and so $\widehat{q}(x)=\frac{x^n-1}{g^*(x)}\mu(x)$ for some $\mu(x)\in\mathbb{F}_2[x]$.

(5) We have that
\begin{eqnarray*}
(\widehat{g}(x)+u\widehat{p_1}(x)+u^2\widehat{p_2}(x))(ua_1(x)+u^2q(x))^*&=&u\widehat{g}(x)a_1^*(x)+u^2\widehat{g}(x)x^{l-m}q^*(x)\\&+&u^2\widehat{p_1}(x)a_1^*(x)=0
\mbox{ mod } \langle x^n-1\rangle. 
\end{eqnarray*}
Since $a_2^*(x)\mid a_1^*(x)$ and $a_2^*(x)\mid q^*(x)$, by part (1) it follows 
$u\widehat{g}(x)a_1^*(x)=u^2\widehat{g}(x)q^*(x)=0$ mod $\langle x^n-1\rangle$.
Consequently
$\widehat{p_1}(x)=\frac{x^n-1}{a_1^*(x)}\nu(x)$  for some $\nu(x)\in\mathbb{F}_2[x]$.

(6)  We can easily see that
{\small
\begin{eqnarray*}
(\widehat{g}(x)+u\widehat{p_1}(x)+u^2\widehat{p_2}(x))(g(x)+up_1(x)+u^2p_2(x))^*
&=&\widehat{g}(x)g^*(x)+u\widehat{g}(x)x^{r-s}p_1^*(x)\\
&+&u^2\widehat{g}(x)x^{r-t}p_2^*(x)
+u\widehat{p_1}(x)g^*(x)\\
&+&u^2\widehat{p_1}(x)x^{r-s}p_1^*(x)+u^2\widehat{p_2}(x)g^*(x)\\
&=&0\mbox{ mod } \langle x^n-1\rangle. 
\end{eqnarray*}}
By part (1) we have that $\widehat{g}(x)g^*(x)=\widehat{g}(x)p_1^*(x)
=\widehat{g}(x)p_2^*(x)=0$ mod $\langle x^n-1\rangle$. Also, part (5) implies $\widehat{p_1}(x)g^*(x)
+\widehat{p_1}(x)p_1^*(x)=0$ mod $\langle x^n-1\rangle$. Therefore $\widehat{p_2}(x)=\frac{x^n-1}{g^*(x)}\eta(x)$ for some $\eta(x)\in\mathbb{F}_2[x]$.
\end{proof}

\section{Conclusion}
We have investigated some properties of generator polynomials of reversible and reversible-complement cyclic codes with even lenght over the ring $\mathbb{F}_2+u\mathbb{F}_2+u^2\mathbb{F}_2$ using the mathematical structure of DNA. We have found a condition under which a reversible cyclic code is reversible-complement. Moreover, we have presented the relationship between the generator polynomials of cyclic codes over $\mathbb{F}_2+u\mathbb{F}_2+u^2\mathbb{F}_2$ and their duals.

\vspace{5mm} \noindent \footnotesize 
\begin{minipage}[b]{10cm}
Hojjat Mostafanasab \\
Department of Mathematics and Applications, \\ 
University of Mohaghegh Ardabili, \\ 
P. O. Box 179, Ardabil, Iran. \\
Email: h.mostafanasab@gmail.com, \hspace{1mm} h.mostafanasab@uma.ac.ir
\end{minipage}\\

\vspace{5mm} \noindent \footnotesize 
\begin{minipage}[b]{10cm}
Ahmad Yousefian Darani \\
Department of Mathematics and Applications, \\ 
University of Mohaghegh Ardabili, \\ 
P. O. Box 179, Ardabil, Iran. \\
Email: youseffian@gmail.com, \hspace{1mm} yousefian@uma.ac.ir
\end{minipage}\\

\end{document}